\documentclass[prl,twocolumn,showpacs]{revtex4}

\usepackage{amsmath}
\usepackage{amsthm}
\usepackage{enumitem}

\begin{document}

\newtheorem{prop}{Proposition}
\newtheorem*{lem}{Lemma}
\newtheorem{cor}{Corollary}

\title{Simple state preparation for contextuality tests 
with few observables}
\author{S. Camalet}
\affiliation{Laboratoire de Physique Th\'eorique 
de la Mati\`ere Condens\'ee, UPMC, CNRS UMR 7600, 
Sorbonne Universit\'es, 4 Place Jussieu, 75252 Paris 
Cedex 05, France}

\begin{abstract}
We consider any noncontextuality inequality, and 
the state preparation scheme which consists in 
performing any von Neumann measurement on 
any initial state. For an inequality which is not always 
satisfied, and Hilbert space dimensions greater than 
a value specified by the inequality, we determine 
necessary and sufficient conditions for the existence 
of observables with which the inequality is violated 
after the preparation process. For an initial 
state with no zero eigenvalues, there are always 
such observables, and which are independent of 
this state.
\end{abstract} 


\maketitle 
  
Quantum mechanics is contextual. Measurement 
outcomes are not predetermined independently of 
the measurements actually performed \cite{B1,KS}. 
This can be revealed with a finite set of observables, 
such that each observable is compatible with some other 
ones, but not with all. When evaluated with 
a noncontextual hidden-variable theory, the correlations 
of the compatible observables, satisfy inequalities, which 
can be violated by quantum systems. Well-known 
examples of such noncontextuality inequalities are 
Clauser-Horne-Shimony-Holt (CHSH) and 
Klyachko-Can-Binicio\u glu-Shumovsky (KCBS) 
inequalities \cite{B2,CHSH,F,KCBS}. Inequalities have 
been obtained, which are disobeyed for any state of 
the considered system. Moreover, this holds for fixed sets 
of observables, independent of the state. 
These inequalities involve 13 dichotomic observables for 
a three-level system \cite{YO,C1,CKB}, and 9 for 
a four-level system \cite{P,M,C2}.  However, contextuality 
can be revealed with less observables, 4 are enough 
for a $d$-level system with $d \ge 4$, using CHSH 
inequality.  

But contextuality tests with few observables, have two 
drawbacks. First, for some states, the noncontextuality 
inequality is satisfied with any set of observables 
obeying the required compatibility relations. Second, 
for the other states, the observables must be chosen 
according to the state, in order to violate the inequality 
\cite{G,ED,KK,XSC,RH}. 
A noncontextuality inequality involves a sum of 
expectation values. Thus, if it cannot be violated for pure 
states, it is always satisfied, and is hence not a proper 
contextuality test. The maximally mixed state of 
the considered system also plays a particular role. 
For three-level systems, an inequality, with 
9 observables, has been found, that can be violated 
for any state except the maximally mixed one 
\cite{KK}. The relation between the mixedness 
of a state, in the sense of majorization \cite{MO,HLP}, 
and its usefulness for revealing quantum contextuality, 
has been clarified in Ref.\cite{RH}. The eigenvalues 
of the state do not alone dictate if a given 
noncontextuality inequality can be disobeyed. 
The dimension of the system Hilbert space is 
an important parameter, since it determines 
the set of potentially accessible observables 
\cite{GBCKL}.  

In this Letter, we show that the two above mentioned 
drawbacks do not mean that violating a given 
noncontextuality inequality necessarily requires 
a very efficient state preparation. We consider the state 
preparation scheme which consists in starting from 
any initial state, and performing any von Neumann 
measurement. When the dimension $d$ of the system 
Hilbert space is greater than a value, determined by 
the inequality, which is, e.g., 4 for CHSH inequality, and 
the inequality is a proper contextuality test, it cannot be 
violated after the preparation measurement, if only if 
a single outcome of this measurement has nonzero 
probability, and the inequality is always satisfied for 
the initial state. This exceptional case is easy to detect, 
and cannot occur if the ranks of the measurement 
projectors are lower than $d$ minus a constant, 
specified by the inequality. This constant is equal to $2$ 
for CHSH inequality, for example. Moreover, when 
the initial state has no zero eigenvalues, the inequality 
is necessarily violated after the preparation process, 
and with observables independent of this state.
 
A noncontextuality inequality, with $N$ dichotomic 
observables $A_k$, reads 
\begin{equation}
\sum_n x_n \big\langle \prod_{k \in {\cal E}_n}  
A_k \big\rangle \le 1 , \label{nci}
\end{equation}
where $\langle \ldots \rangle=$Tr$(\rho \ldots)$ 
denotes the average with respect to the quantum 
state $\rho$. The subsets $
{\cal E}_n \subset \{ 1, \ldots , N \}$, are such 
that $[A_k,A_l]=0$ when $k$ and $l$ both belong 
to ${\cal E}_n$. The number of terms in the sum 
and the coefficients $x_n$ depend on the inequality 
considered. A noncontextuality inequality is always 
satisfied when the observables $A_k$ are replaced 
by classical random variables $a_k=\pm 1$, and 
the average is evaluated with respect to a probability 
distribution of these variables. Moreover, 
the coefficients $x_n$ are such that the maximum 
value of the left-hand side of eq.\eqref{nci} with 
classical random variables, is $1$. Thus, a violation 
of inequality \eqref{nci} clearly indicates that 
the obtained value cannot be accounted for by 
a noncontextual hidden-variable theory. 

The states $\rho$ of a $d$-level system, for which 
a given inequality \eqref{nci} is violated with 
appropriate observables $A_k$, are determined 
by the function C$_d$ defined as 
\begin{equation}
\mathrm{C}_d(\rho) = \max_{{\bf A} \in {\cal A}_d} 
\mathrm{Tr} \big[ \rho T({\bf A}) \big] , \label{Cdef}
\end{equation}
where ${\bf A}=( A_1, \ldots , A_N )$, 
$T({\bf A})=\sum_n x_n \prod_{k \in {\cal E}_n}  A_k$, 
and ${\cal A}_d$ denotes the set of all ${\bf A}$ 
consisting of dichotomic observables 
$A_k$, of the $d$-level system, 
which obey $[A_k,A_l]=0$ for $k,l \in {\cal E}_n$.
Note that this definition depends on the dimension $d$ 
of the considered Hilbert space. By construction, 
for a state $\rho$ such that $\mathrm{C}_d(\rho) \le 1$, 
inequality \eqref{nci} is satisfied with any dichotomic 
observables $A_k$ obeying the required commutation 
relations.

It results directly from the definition \eqref{Cdef} that 
the function C$_d$ is continuous, invariant under 
unitary transformations of $\rho$, and convex.
\begin{prop}\label{bp} 
For any states $\rho_m$, unitary operator $U$, and 
probabilities $p_m$ such that $\sum_m p_m=1$, 
\begin{enumerate}[label=\roman*)]
\item $\big| \mathrm{C}_d(\rho_1)
-\mathrm{C}_d(\rho_2) \big| \le \sqrt{d} \sum_n |x_n| 
\mathrm{Tr} \big[ (\rho_1-\rho_2)^2 \big]^{1/2}   
\label{cont}$,
\item $\mathrm{C}_d(U\rho_1 U^\dag)
=\mathrm{C}_d(\rho_1)  \label{Uinv}$,
\item $\mathrm{C}_d\big(\sum_m p_m \rho_m \big) 
\le \sum_m p_m \mathrm{C}_d(\rho_m) \label{conv}$. 
\end{enumerate}
\end{prop}
\begin{proof}
i) It follows from the Cauchy-Schwarz inequality that 
$\mathrm{Tr} [  \omega T({\bf A}) ]^2 \le 
\mathrm{Tr} (\omega ^2 ) \sum_{i=1}^d t_i^2$, 
where $\omega=\rho_1-\rho_2$, and $t_i$ denotes 
the eigenvalues of $T({\bf A})$. Since 
$|\langle \prod_k  A_k \rangle|  \le 1$ for any state $\rho$, 
$|t_i| \le \sum_n |x_n|$. Consequently, for any 
${\bf A} \in {\cal A}_d$, 
$\mathrm{Tr} [ \rho_{1/2} T({\bf A}) ] 
\le \mathrm{C}_d(\rho_{2/1}) + 
[d \mathrm{Tr} (\omega ^2 )]^{1/2} \sum_n |x_n|$. 
Maximizing over ${\bf A}$ completes the proof of 
point \ref{cont}.

ii) $\mathrm{Tr} [  U\rho_1 U^\dag T({\bf A}) ] 
= \mathrm{Tr} [  \rho_1 T({\bf B}) ]$ where 
$B_k=U^\dag A_k U$. For ${\bf A} \in {\cal A}_d$, 
the observables $B_k$ satisfy the commutation relations 
$[B_k,B_l]=U^\dag[A_k,A_l] U = 0$ for $k,l \in {\cal E}_n$, 
and are dichotomic, since $B_k^2=U^\dag A_k^2 U=I_d$ 
where $I_d$ is the $d$-dimensional identity operator. 
Consequently, ${\bf B}$ belongs to ${\cal A}_d$. Thus, 
the above equality yields 
$\mathrm{Tr} [  U\rho_1 U^\dag T({\bf A}) ] 
\le \mathrm{C}_d(\rho_1)$ for any ${\bf A} \in {\cal A}_d$, 
and hence $\mathrm{C}_d(U\rho_1 U^\dag) \le 
\mathrm{C}_d(\rho_1)$. Since this inequality is valid 
for any $\rho_1$ and $U$, the equality holds for 
any $\rho_1$ and $U$.

iii) By linearity of the trace, 
$\mathrm{Tr} [  \sum_m p_m \rho_m T({\bf A})  ]
\le \sum_m p_m \mathrm{C}_d(\rho_m)$ 
for any ${\bf A} \in {\cal A}_d$, which proves \ref{conv}.
\end{proof}
From point \ref{cont} of proposition \ref{bp}, it ensues that, 
any state $\rho$ such that $\mathrm{C}_d(\rho)>1$, has 
a neighborhood of states which can violate eq.\eqref{nci}. 
Thus, no noncontextuality inequality can be disobeyed 
only for pure states. Points \ref{Uinv} and \ref{conv} 
show that applying unitary transformations to states $\rho$ 
such that $\mathrm{C}_d(\rho) \le 1$, or preparing 
statistical mixtures of such states, cannot lead to 
a violation of inequality \eqref{nci}.  Another result 
of point \ref{Uinv} is that $\mathrm{C}_d(\rho)$ depends 
only on the eigenvalues of the state $\rho$. 

The convexity and invariance under unitary 
transformations of C$_d$ have the following consequence 
for positive operator-valued measurements. From now on, 
we use the notation
\begin{equation}
\mathrm{C}_d(F,\rho) =\mathrm{C}_d
\big[F\rho F^{\dag}/\mathrm{Tr}(F^{\dag}F\rho)\big] ,
\end{equation}
where $\rho$ is a state, and $F$ any operator such that 
$F\rho F^{\dag} \ne 0$, of a $d$-level system.
\begin{cor}
For any state $\rho$, and operators $F_m$ such 
that $\sum_m F^\dag_m F^{\phantom{\dag}}_m=I_d$, 
of a $d$-level system,  $\mathrm{C}_d(\rho) \le 
\max_{m\in {\cal E}} \mathrm{C}_d(F_m,\rho)$,
where ${\cal E}=\{ m : 
F^{\phantom{\dag}}_m \rho F^\dag_m \ne 0 \}$. 
\end{cor}
\begin{proof} 
There are unitary operators $U_m$ such that 
$\rho=\sum_{m\in {\cal E}} p_m U^{\phantom{\dag}}_m
\rho_m U_m^\dag$ where 
$p_m=\mathrm{Tr}(F^\dag_mF^{\phantom{\dag}}_m \rho)$ 
and $\rho_m=F^{\phantom{\dag}}_m \rho F^\dag_m/p_m$ 
\cite{QMT}. Thus, with proposition \ref{bp}, 
$\mathrm{C}_d(\rho) \le \sum_{m \in {\cal E}} p_m 
\mathrm{C}_d( \rho_m)$, which leads, 
with $\sum_{m \in {\cal E}} p_m=1$, to the result.
\end{proof}
In other words, for any measurement, at least one 
resulting state $\rho_m$ gives a value of $\mathrm{C}_d$ 
which can exceed that of the initial state. However, 
this obviously does not guarantee that inequality \eqref{nci} 
can be violated. In the following, we show conditions 
under which this is the case for von Neumann 
measurements, i.e., if the operators $F_m$ are 
projectors. 

Below, we make use of the majorization relation, which is 
defined as follows. Consider two real $d$-component 
vectors $\bf a$ and $\bf b$, and the vectors 
$\bf a^\downarrow$ and $\bf b^\downarrow$ obtained 
from $\bf a$ and $\bf b$, respectively, by rearranging 
their components in decreasing order, i.e., 
$a_i^\downarrow \ge a_{i+1}^\downarrow$. It is said 
that $\bf a$ majorizes $\bf b$, denoted $\bf a \succ \bf b$, 
iff, for $j=1, \ldots, d$, $\sum_{i=1}^j a_i^\downarrow 
\ge \sum_{i=1}^j b_i^\downarrow$, with equality for $j=d$. 
For density matrices, $\rho_1 \succ \rho_2$ iff 
$\boldsymbol \lambda (\rho_1) 
\succ \boldsymbol \lambda (\rho_2)$, 
where the spectrum $\boldsymbol \lambda (A)$ is 
the vector made up of the eigenvalues of the Hermitian 
operator $A$, in decreasing order \cite{MO,HLP}. 
The majorization relation is generalized to states 
of systems of different sizes, by extending with zeros 
the spectrum with less eigenvalues. The next 
proposition will be proved using the following lemma.
\begin{lem}
Consider three real $d$-component 
vectors $\bf a$, $\bf b$ and $\bf c$. 
If ${\bf a} \succ {\bf b}$ 
then ${\bf b} \cdot {\bf c} \le {\bf a}^\downarrow 
\cdot {\bf c}^\downarrow \label{lem}$, 
where ${\bf a} \cdot {\bf b}=\sum_{i=1}^d a_i b_i$. 
\end{lem}
\begin{proof} 
It is already known that 
${\bf b} \cdot {\bf c} \le {\bf b}^\downarrow 
\cdot {\bf c}^\downarrow$ \cite{MO,HLP}.
We define $R_j=\sum_{i=1}^j 
(b^\downarrow_i- a^\downarrow_i)$ for $j=1,\ldots,d$. 
Since ${\bf a} \succ {\bf b}$, $R_j \le 0$ and $R_d=0$.
Thus, $({\bf b}^\downarrow  - {\bf a}^\downarrow)
\cdot {\bf c}^\downarrow =\sum_{j=1}^{d-1}
(c^\downarrow_j-c^\downarrow_{j+1})R_j \le 0$. 
\end{proof}

To investigate the influence of the Hilbert space dimension 
$d$, we define the application 
${\cal G}: {\bf A'} \mapsto {\bf A}$ as follows. As mentioned 
above, for any noncontextuality inequality \eqref{nci}, 
there is $(a_1, \ldots, a_N)$ such that $a_k = \pm 1$ and 
$\sum_n x_n  \prod_{k \in {\cal E}_n} a_k = 1$. 
For any ${\bf A'} \in {\cal A}_{d'}$, 
${\bf A}=( A_1, \ldots , A_N ) \in {\cal A}_d$, is given 
by $\langle \tilde \imath | A_k |\tilde \jmath \rangle=
\langle i | A'_k | j \rangle$, for $d=d'$, and by
\begin{equation}
A_k = \sum_{i,j=1}^{d'} \langle i | A'_k | j \rangle 
|\tilde \imath \rangle \langle \tilde \jmath |
+a_k \sum_{i=d'+1}^d 
|\tilde \imath \rangle \langle \tilde \imath | , \label{Gfun}
\end{equation}
for $d>d'$, where $\{ | i \rangle \}_{i=1}^{d'}$ and 
$\{ | \tilde \imath \rangle \}_{i=1}^d$ are orthonormal 
bases of the considered Hilbert spaces. With matrix 
representations of the observables $A_k$ and $A'_k$, 
it is straightforward to show that, 
when ${\bf A'} \in {\cal A}_{d'}$, the observables $A_k$ 
are dichotomic and obey the required commutation 
relations, and that the spectrum of $T({\bf A})$ consists 
of the $d'$ eigenvalues $\lambda_i[T({\bf A'})]$ and 
of $d-d'$ ones. Using the lemma and function ${\cal G}$, 
the following can be shown.
\begin{prop}\label{Sc}
Consider a state $\rho$ of a $d$-level system, and a state 
$\rho'$ of a $d'$-level system. If $\rho \succ \rho'$ and 
$d \ge d'$ then 
$\mathrm{C}_d(\rho) \ge \mathrm{C}_{d'}(\rho')$ . 
\label{Scg}
\end{prop}
\begin{proof} 
We first prove that 
\begin{equation}
\mathrm{C}_d(\rho) = \max_{{\bf t} \in \Lambda_d} 
[{\bf t} \cdot {\boldsymbol \lambda} (\rho) ] \label{nf}
\end{equation}
where $\Lambda_d=\{ \boldsymbol \lambda[T({\bf A})] 
: {\bf A} \in {\cal A}_d \}$. 
For that purpose, we write 
$\mathrm{Tr} [ \rho T({\bf A}) ]=\sum_{i=1}^d t_i p_i$ 
where $p_i=\langle i |\rho |i \rangle$, $|i \rangle$ denotes 
the eigenvectors of $T({\bf A})$, and ${\bf t}$ its spectrum. 
The Schur-Horn theorem gives 
$\boldsymbol \lambda(\rho) \succ (p_1, \ldots, p_d)$ 
\cite{MO}. Thus, using the lemma, we obtain 
$\mathrm{Tr} [ \rho T({\bf A}) ] \le 
{\bf t} \cdot {\boldsymbol \lambda} (\rho)$, 
which results in 
$\mathrm{C}_d(\rho) \le \max_{{\bf t} \in \Lambda_d} 
[{\bf t} \cdot {\boldsymbol \lambda} (\rho)]$. 

Consider ${\bf t} \in \Lambda_d$. By definition of 
$\Lambda_d$, there is $\bf A \in {\cal A}_d$ such that 
$\boldsymbol \lambda[T({\bf A})]={\bf t}$. Consider 
${\bf B}$ defined by $B_k=U^\dag A_k U$ where $U$ 
is any unitary operator. ${\bf B}$ belongs to ${\cal A}_d$ 
(see proof of point \ref{Uinv} of proposition \ref{bp}). 
Moreover, 
$T({\bf B})=\sum_{i=1}^d t_i U^\dag |i \rangle\langle i |U$. 
Therefore, there is ${\bf \tilde A} \in {\cal A}_d$ such that 
the spectrum of $T({\bf \tilde A})$ is ${\bf t}$, and 
its eigenvectors are identical to those of $\rho$. Hence, 
$\sum_{i=1}^d t_i \lambda_i(\rho)
=\mathrm{Tr} [\rho T({\bf \tilde A}) ] \le \mathrm{C}_d(\rho)$,
which gives the second inequality required to prove 
eq.\eqref{nf}.

Consider a state ${\tilde \rho '}$ of a $d$-level system, 
with the same nonzero eigenvalues than $\rho'$. 
Since $\rho \succ \rho'$, 
$\boldsymbol \lambda (\rho) \succ 
\boldsymbol \lambda (\tilde \rho')$. Thus, using the lemma 
and the form \eqref{nf}, we have 
${\bf t} \cdot {\boldsymbol \lambda} (\tilde \rho') 
\le \mathrm{C}_d(\rho)$ for any ${\bf t} \in \Lambda_d$. 

For any ${\bf A'} \in {\cal A}_{d'}$, expression \eqref{Gfun} 
gives ${\bf A} \in {\cal A}_{d}$ such that 
the components of ${\bf t}=\boldsymbol \lambda[T({\bf A})]$ 
are the $d'$ eigenvalues $t'_i=\lambda_i[T({\bf A'})]$, 
and $d-d'$ ones, arranged in decreasing order. Thus, 
$t_i = t'_i$ if $t'_i \ge 1$, and $t_i = 1$ or $t'_{i-j}$ where 
$j \ge d-d'$, if $t'_i < 1$. So, for $i \le d'$, $t'_i \le  t_i$, 
and hence $\mathrm{Tr} [ \rho' T({\bf A'}) ] \le
\sum_{i=1}^{d'} t'_i \lambda_i(\rho') 
\le \sum_{i=1}^d  t_i \lambda_i(\tilde \rho')$, which, 
together with the above inequality, leads to the result.
\end{proof}
For a von Neumann measurement, the resulting states 
$\rho_m =\Pi_m \rho \Pi_m/\mathrm{Tr}(\Pi_m \rho)$ 
where $\Pi_m$ are projectors and $\rho$ is the initial state, 
have vanishing eigenvalues. To study their ability 
to violate inequality \eqref{nci}, it is convenient to define
\begin{equation}
\mathrm{C}_d^{(r)} = \max_{{\bf t} \in \Lambda_d} 
\sum_{i=1}^r t_i/r ,  \label{Cdr}
\end{equation}
where $r \le d$. Noting that 
$\mathrm{C}_d^{(r)} =\mathrm{C}_d(\Pi/r)$ 
where $\Pi$ is any rank-$r$ projector, see eq.\eqref{nf}, 
and using proposition \ref{Sc}, the following properties 
of $\mathrm{C}_d^{(r)}$ and $\mathrm{C}_d$ can be 
proved. Any density matrix $\rho$ of rank $r$, 
satisfies $\rho \succ \Pi/r$, and hence 
$\mathrm{C}_d(\rho) \ge \mathrm{C}_d^{(r)}$. 
Consequently, if $\mathrm{C}_d^{(r)}>1$, 
a $d$-level system in such a state $\rho$, can violate
inequality \eqref{nci}. Since $\Pi/r \succ \Pi'/r'$ 
where $r'\ge r$ and $\Pi'$ is a rank-$r'$ projector, 
$\mathrm{C}_d^{(r)}$ decreases as $r$ increases, and 
increases with $d$. 

The function C$_d$ reaches 
its maximum $\mathrm{C}_d^{(1)}$ for pure states, 
which majorize any other state, and its minimum 
$\mathrm{C}_d^{(d)}$ for the maximally mixed state 
$I_d/d$, which is majorized by any state of rank not 
larger than $d$. Thus, these two extreme values 
determine, for a $d$-level system, whether inequality 
\eqref{nci} can be disobeyed or not, for all states or not. 
If $\mathrm{C}_d^{(1)} \le 1$, eq.\eqref{nci} is satisfied 
with any observables $A_k$ obeying the required 
commutation relations, for any system state $\rho$. 
It is then not a proper contextuality test for dimension 
$d$. If $\mathrm{C}_d^{(d)} > 1$, inequality \eqref{nci} 
can be violated for any state $\rho$, but it may remain 
necessary to choose the observables $A_k$ according 
to $\rho$. 
If $\mathrm{C}_d^{(d)} \le 1< \mathrm{C}_d^{(1)}$, 
observables $A_k$ can be found to disobey eq.\eqref{nci} 
or not, depending on the spectrum of $\rho$. 

Since $\mathrm{C}_d^{(1)}$ increases with $d$, 
if inequality \eqref{nci} is a contextuality test for 
a dimension $d'$, it is also so for dimensions $d \ge d'$. 
Relations \eqref{Gfun} lead to
$\mathrm{C}_d^{(d)} \ge 1
+ (\mathrm{C}_{d'}^{(d')}-1)d'/d$,
for $d \ge d'$. Thus, if eq.\eqref{nci} can be violated 
for any state of a $d'$-level system, this is also the case 
for a larger system. The increase with $d$ of 
$\mathrm{C}_d^{(r)}$ obviously does not guarantee 
that it exceeds $1$ for large enough $d$. Below, we show 
that this is actually the case, under the only assumption 
that inequality \eqref{nci} is not always satisfied, and draw 
consequences for von Neumann preparation 
measurements.
\begin{prop}
Consider a state $\rho$, and projectors $\Pi_m$ 
such that $\sum_m \Pi_m = I_d$, their ranks are not 
larger than $r$, and the rank of $\Pi_1$ is $r$, of 
a $d$-level system, and define ${\cal E}$ the set of $m$ 
such that $\mathrm{Tr}(\Pi_m \rho) \ne 0$. Assume there is 
$d'$ such that $\mathrm{C}_{d'}^{(1)} > 1$, and $d \ge d'$.
\begin{enumerate}[label=\roman*)]
\item If $r \le d-d'+1$, then 
$\mathrm{C}_d (\Pi_m,\rho) > 1$ 
for any $m \in {\cal E}$.\label{mi}\item If $r > d-d'+1$, 
$d \ge 2d'-3$, and $\mathrm{Tr}(\Pi_1 \rho) \ne 1$, then 
$\mathrm{C}_d (\Pi_m,\rho) > 1$ for at least one 
$m \in {\cal E}$. 
\label{mii}
\end{enumerate}
\label{m}
\end{prop}
\begin{proof} 
We first show that $\mathrm{C}_d^{(r)} > 1$ 
for $r \le d-d'+1$. 
Since $\mathrm{C}_{d'}^{(1)} > 1$, there is 
${\bf A'} \in {\cal A}_{d'}$ such that $t'_1>1$, where 
${\bf t'}={\boldsymbol \lambda}[T({\bf A'})]$, see 
eq.\eqref{Cdr}. Consider ${\bf A}\in {\cal A}_d$, 
following from eq.\eqref{Gfun}, and denote by ${\bf t}$ 
the spectrum of $T({\bf A})$. We have $t_1=t'_1$ 
and $t_i \ge 1$ for $i \le d-d'+1$. Consequently, 
$\sum_{i=1}^r t_i/r \ge 1+(t'_1-1)/r>1$. 

\noindent We define 
$\rho_m=\Pi_m \rho \Pi_m/\mathrm{Tr}(\Pi_m \rho)$ 
for any $m \in {\cal E}$.

i) Since $\rho_m \succ \Pi_m/r_m$ where $r_m$ 
is the rank of $\Pi_m$, and $r_m \le r$, proposition 
\ref{Scg} gives 
$\mathrm{C}_d (\rho_m) \ge C_d^{(r_m)}\ge C_d^{(r)}$. 
So, using the above result, we get 
$\mathrm{C}_d (\rho_m)>1$.

ii) There is  $m \ne 1$ such that 
$\mathrm{Tr}(\Pi_{m} \rho) \ne 0$. The rank $r_m$ of 
$\rho_{m}$ is not larger than $d-r \le d-d'+1$, and 
hence, $\mathrm{C}_d (\rho_{m})\ge C_d^{(r_m)} > 1$.
\end{proof}
For dimensions $d \ge 2d'-3$ where $d'$ is such that 
$\mathrm{C}_{d'}^{(1)} > 1$, it results from 
proposition \ref{m} that inequality \eqref{nci} 
cannot be violated after the preparation measurement, 
only if  a single outcome of this measurement has 
nonzero probability, and $\mathrm{C}_d (\rho) \le 1$ 
where $\rho$ is the initial state. This last condition 
comes from the fact that the sole 
post-measurement state is equal to $\rho$.
\begin{cor}
Consider a state $\rho$, and projectors $\Pi_m$ 
such that $\sum_m \Pi_m = I_d$, of a $d$-level system. 
Assume there is $d'$ such that 
$\mathrm{C}_{d'}^{(1)} > 1$, and $d \ge 2d'-3$.

$\mathrm{C}_d (\Pi_m,\rho) \le 1$ for all $\Pi_m$ such 
that $\mathrm{Tr}(\Pi_m \rho) \ne 0$ iff 
$\mathrm{Tr}(\Pi_m \rho)=1$ for one $m$, 
and $\mathrm{C}_d (\rho) \le 1$.
\end{cor}
\begin{proof}
If $\mathrm{Tr}(\Pi_m \rho)=1$, then
$\mathrm{Tr}(\Pi_{m'} \rho)=0$ 
for any $m' \ne m$, and
$\rho=\sum_{m',m''} \Pi_{m'} \rho \Pi_{m''}
=\Pi_m \rho \Pi_m$.

If $\mathrm{C}_d (\Pi_{m'},\rho) \le 1$ for all appropriate 
$\Pi_{m'}$, then, due to proposition \ref{m}, 
$\mathrm{Tr}(\Pi_m \rho)=1$ where $\Pi_m$ is 
the projector of largest rank, and thus 
$\rho=\Pi_m \rho \Pi_m$. 
\end{proof}
Proposition \ref{m} concerns the post-measurement 
states $\rho_m$, and the possibility to violate inequality 
\eqref{nci} after one of them was selected. If, on 
the contrary, the measurement is unread, the state 
of the system after it, is 
$\rho' = \sum_m \Pi_m \rho \Pi_m$, 
which obeys $\rho \succ \rho'$, due to quantum 
Hardy-Littlewood-P{\'o}lya theorem \cite{QMT}. 
It follows from proposition \ref{Sc}, that eq.\eqref{nci} is 
always obeyed for $\rho'$ if $\mathrm{C}_{d} (\rho) \le 1$. 
Interesting measurements are dichotomic ones 
with projectors of ranks $d/2$ for even $d$, 
and $(d\pm1)/2$ for odd $d$. They are the most 
inefficient in the sense that there is a projector $\Pi_m$ 
of lower rank for all the other measurements. 
For these measurements, inequality \eqref{nci} can be 
disobeyed for both resulting states, provided 
$d \ge 2d'-2$. The smaller is the dimension $d'$, 
the less demanding are the conditions in proposition 
\ref{m}. The minimum $d'$ such that 
$\mathrm{C}_{d'}^{(1)} > 1$ is $4$ for CHSH inequality 
\cite{CHSH}, and $3$ for KCBS inequality \cite{KCBS}. 

Proposition \ref{m} ensures that, for any state $\rho$ 
with no zero eigenvalues, and any von Neumann 
measurement, there are observables $A_k$ such that 
inequality \eqref{nci} is violated for a post-measurement 
state $\rho_m$, provided $d \ge 2d'-3$. Such 
observables $A_k$ depend a priori on $\rho$. We show 
below that some of them are determined only by 
the considered measurement and inequality.
\begin{prop}
Consider projectors $\Pi_m$ such that 
$\sum_m \Pi_m = I_d$, of a $d$-level system.

If there is $d'$ such that $\mathrm{C}_{d'}^{(1)} > 1$, and 
$d \ge 2d'-3$, then there are ${\bf A} \in {\cal A}_d$ 
and a projector $\Pi_m$, such that 
$\mathrm{Tr} [\rho_mT({\bf A})] > 1$ where 
$\rho_m=\Pi_m \rho \Pi_m/\mathrm{Tr}(\Pi_m \rho)$, 
for all states $\rho$ with no zero eigenvalues, 
of the $d$-level system.
\label{m2}
\end{prop}
\begin{proof} 
There is at least one projector $\Pi_m$ of rank 
$r \le d/2$. Denote by $| \phi \rangle$ one of 
its eigenvectors with eigenvalue $1$. Since 
$\mathrm{C}_{d'}^{(1)} > 1$, there is 
${\bf B'} \in {\cal A}_{d'}$ such that $t'_1>1$, where 
${\bf t'}={\boldsymbol \lambda}[T({\bf B'})]$. 
Consider ${\bf B}\in {\cal A}_d$, following from 
eq.\eqref{Gfun}, and define 
${\bf t}={\boldsymbol \lambda}[T({\bf B})]$. 
We have $t_1=t'_1$, and, since $r \le d-d' +1$, 
$t_i \ge 1$ for $i \le r$. There is ${\bf A} \in {\cal A}_d$ 
such that the spectrum of $T({\bf A})$ is ${\bf t}$, its first 
$r$ eigenvectors $|i\rangle$ obey 
$\Pi_m|i\rangle=|i\rangle$, and $|1\rangle=|\phi \rangle$ 
(see proof of proposition \ref{Sc}). 
$\mathrm{Tr}(\Pi_m \rho) \ge r \lambda$, 
and $\langle \phi |\rho |\phi \rangle \ge \lambda$, 
where $\lambda= \min_j\lambda_j(\rho) > 0$.
Finally, the above results lead to 
$\mathrm{Tr} [ \rho_m T({\bf A}) ] \ge 1+ (t'_1-1) p$ 
where $p=\langle \phi |\rho |\phi \rangle/
\mathrm{Tr}(\Pi_m \rho) > 0$. 
\end{proof}
For an initial state 
$\rho=\sum_{i=1}^d\lambda_i(\rho)|i\rangle\langle i|$ 
of a $D$-level system, of rank $d < D$, proposition 
\ref{m2} holds for the subspace spanned by 
$\{ |i \rangle \}_{i=1}^d$, if $d$ is large enough. Thus, 
a violation of inequality \eqref{nci} can be achieved 
knowing only the subspace corresponding to the zero, 
or very small, eigenvalues of $\rho$, and choosing 
the preparation measurement, and observables $A_k$, 
accordingly. 

In summary, we have studied the possibility of preparing 
a state that violates a given noncontextuality inequality, 
by performing any von Neumann measurement on any 
initial state. For a large enough system, and an inequality 
which is not always satisfied, we have determined 
necessary and sufficient conditions for the existence 
of observables with which the inequality is violated for 
a state resulting from the preparation measurement. 
For an initial state with no zero eigenvalues, there are 
always such observables, and which do not depend on 
this state. A natural extension of this work is to consider, 
for the preparation stage, general measurements, for 
which only a partial result has been obtained.

\end{document}